\documentclass[prodmode,permissions]{acmsmall-ec16}
\usepackage{latexsym}
\usepackage{amsmath,amsfonts,amssymb,bbm} 
\usepackage[numbers,sort&compress]{natbib} 
\usepackage{subfigure}
\usepackage{float}
\usepackage{graphicx,color}
\usepackage[dvipsnames]{xcolor}
\usepackage{cite,url}
\usepackage{dsfont}
\usepackage{pst-all}
\usepackage{pst-node}
\usepackage{graphicx}
\usepackage{epsfig}
\usepackage{algorithm}
\usepackage[noend]{algorithmic}
\usepackage{multirow}
\definecolor{WildStrawberry}{RGB}{255,67,164}
\definecolor{DarkGreen}{rgb}{0,0.5,0.0}
\newtheorem{claim}{Claim}[section]
\newcommand{\be}{\begin{equation}}
\newcommand{\ee}{\end{equation}}
\newcommand{\argmin}{\mathop{\rm argmin}}
\newcommand{\argmax}{\mathop{\rm argmax}}
\newcommand{\AutoAdjust}[3]{\mathchoice{ \left #1 #2  \right #3}{#1 #2 #3}{#1 #2 #3}{#1 #2 #3} }
\newcommand{\Xcomment}[1]{{}}

\newcommand{\InBrackets}[1]{\AutoAdjust{[}{#1}{]}}
\newcommand{\Ex}[2][]{\operatorname{\mathbf E}_{#1}\InBrackets{#2}}

\newcommand{\Prx}[2][]{\operatorname{\mathbf{Pr}}_{#1}\InBrackets{#2}}

\newcommand{\eqdef}{\stackrel{\textrm{def}}{=}}

\newcommand{\R}{\mathbb{R}}

\newcommand{\eps}{\varepsilon}

\newcommand{\Path}{\text{Path}}

\newcommand{\dis}[2]{d(#1,#2)}
\newcommand{\disti}[2][v_i]{\dis{#1}{#2}}
\newcommand{\dtot}[1][s]{\disti[#1]{t}}
\newcommand{\disop}[1][o]{d(#1)}
\newcommand{\nsteps}{n}

\newcommand{\opt}{\mathrm{OPT}}
\newcommand{\rev}{\mathrm{cost}}
\newcommand{\revi}[1][i]{\rev_{#1}}
\newcommand{\revb}{\overline{\rev}}
\newcommand{\revbi}[1][i]{\revb_{#1}}

\newcommand{\thresh}{\biasi[0]}

\newcommand{\dist}{\ensuremath{\mathcal{F}}}
\newcommand{\distb}{\ensuremath{\bar{\mathcal{F}}}}

\newcommand{\fdist}{\ensuremath{\mathfrak{F}}}

\newcommand{\val}{v}
\newcommand{\vali}[1][i]{\val_{#1}}

\newcommand{\bias}{b}
\newcommand{\biasi}[1][i]{{\bias_{#1}}}

\newcommand{\biasz}{b^{*}}

\begin{document}
\title{Procrastination with Variable Present Bias}
\author{
Nick Gravin\affil{Massachusetts Institute of Technology; ngravin@mit.edu}
Nicole Immorlica\affil{Microsoft Research; nicimm@gmail.com}
Brendan Lucier \affil{Microsoft Research; brlucier@microsoft.com}
Emmanouil Pountourakis\affil{Northwestern University; manolis@u.northwestern.edu}
}
\begin{abstract}
Individuals working towards a goal often exhibit time inconsistent behavior, making plans and then failing to follow through. One well-known model of such behavioral anomalies is present-bias discounting: individuals over-weight present costs by a bias factor.  This model explains many time-inconsistent behaviors, but can make stark predictions in many settings: individuals either follow the most efficient plan for reaching their goal or procrastinate indefinitely.

We propose a modification in which the present-bias parameter can vary over time, drawn independently each step from a fixed distribution. 
Following Kleinberg and Oren (2014), we use a weighted {\it task graph} to model  task planning,
and measure the cost of procrastination as the relative expected cost of the chosen path versus the optimal path.  We use a novel connection to optimal pricing theory to describe the structure of the worst-case task graph for any present-bias distribution.  We then leverage this structure to derive conditions on the bias distribution under which the worst-case ratio is exponential (in time) or constant.
We also examine conditions on the task graph 
that lead to improved procrastination ratios: graphs with a uniformly bounded distance to the goal, and graphs in which the distance to the goal monotonically decreases on any path. 

\end{abstract}
\maketitle
\section{Introduction}
\label{sec:intro}

Intertemporal tradeoffs -- in which an individual incurs short-term costs to achieve long-term goals -- have significant implications for health, wealth, and educational outcomes.  A student may complete a homework assignment due next week instead of partying with friends.  A marathoner may decide to go for a morning training run instead of sleeping in.  A skier with a season pass may choose to buy skies instead of renting them each weekend, if it is less costly to do so.  Long-term goals like these often involve many sacrifices over time, requiring the individual to form a consistent plan of future action.  Anecdotally, however, individuals often have trouble sticking to such plans.  This is supported empirically by the observation that experimental subjects discount future payoffs at a non-constant rate.  When asked how much money a subject would require in a month/year/decade to offset \$15 today, the median response was 
\$20/\$50/\$100~\cite{Thaler81}. Since these responses are not consistent with any constant discount factor, a subject with these preferences might plan today to invest \$20 next month to gain \$100 in a decade, but then fail to follow through next month. 

These observations have led to an extensive line of work in behavioral economics, proposing behavior models that allow divergence between plan and action.  The economic theory of {\em hyperbolic discounting} is one such model developed to explain these time inconsistencies.  A hyperbolic discount function assigns weights to future costs.  At each point in time, an individual forms a plan which minimizes total costs calculated according to the hyperbolic discounting function. A special case of hyperbolic discounting, developed by Akerloff~\cite{Akerlof}, is {\em present-bias discounting}. In its simplest form, present-bias discounting predicts that individuals will inflate any costs incurred at the present moment by a model parameter $\bias>1$, and leave all future costs unweighted.  

As an example illustrating how present bias can lead to procrastination, consider a student who must complete a homework assignment.\footnote{This example is a minor reformulation of an example due to Akerloff.}  The assignment, which is based on today's lecture, is due next week.  The homework takes one night to complete, and there is a cost (in mental effort and missed social opportunities) to complete the assignment immediately after the lecture.  Each day that goes by, the student's recollection of lecture becomes dimmer, and the cost for completing the assignment increases.  In this scenario, it is clearly optimal, in terms of minimizing induced costs, to complete the homework assignment the day of the lecture.  However, if the student inflates present-day costs, he may {\it perceive} that it is less costly to complete the assignment tomorrow than today.  Facing a similar decision tomorrow, he again perceives that procrastinating is less costly than completing the assignment immediately.  This procrastination behavior persists, and the student doesn't complete the assignment until the last day, spending substantially more than the optimal plan. 

The preceding example contains a very stark prediction: the student will either procrastinate indefinitely, if present-biased, or not at all.  As noted in prior work of Kleinberg and Oren~\cite{KleinbergO14}, this dichotomy persists even in the face of exponentially increasing costs, causing a present-biased student to spend exponentially more than necessary.  In reality, one should expect a more nuanced behavior -- any given individual might procrastinate for some amount of time, but eventually even the most lackadaisical students will exert effort to avoid future penalties. Such variation in procrastination behavior can arise from variability in the present-bias parameter.  In this work, we assume that the agent has a present-bias parameter $\biasi$ drawn on each day $i$ independently at random from a distribution $\dist$ supported on $[0,1)$.  For the preceding example, if the student's present bias parameter has a constant probability of taking value $1$, then the student will delay for only a constant number of days, in expectation, before completing their assignment.  We ask: how does variability in present-bias change procrastination behavior?  Are there structured choices, e.g., homework late policies, that guide individuals away from costly procrastination?

To study these questions, we adopt a graph-theoretic framework for task completion proposed by Kleinberg and Oren~\cite{KleinbergO14}. In this framework, an agent with a goal $\nsteps$ days in the future traverses through a weighted task graph $G$ starting from node $s$ and ending at node $t$.  On each day $i$, if the agent's current present bias is $\biasi$ and current state is $v$, then the weight of the first edge along a $v-t$ path is multiplied by an extra factor of $\biasi$, while all the remaining edges are evaluated according to their true weights. The agent chooses a minimal-weight path, under this distortion, and follows the first edge of that path.  For a given task graph $G$ and present-bias distribution $\dist$, the {\em procrastination ratio} is the ratio between the expected total weight of the traversed path to the weight of the initial shortest path. We are interested in bounding the procrastination ratio as a function of $\nsteps$, the number of days until the deadline for the task.\footnote{Kleinberg and Oren~\cite{KleinbergO14} bounded the procrastination ratio with respect to the total number of nodes in the graph, rather than the path length $\nsteps$, but these are essentially equivalent; see Section~\ref{sec:model}.}

Under this model, we can state our question more concretely: are there natural conditions that guarantee a low procrastination ratio?  Recall that in the student  example, if there is a constant probability of bias $b_i=1$, the student will complete the assignment in a constant number of days in expectation.  As it turns out, this is not sufficient to guarantee a sub-exponential procrastination ratio (see Section~\ref{sec:model} for more details).  Our first result shows that, for any distribution, the graph with maximal procrastination ratio has exactly this form: on any day, the agent may either complete the task for a growing cost or procrastinate for free.  The proof of this result involves a connection to optimal pricing theory: the problem of constructing a worst-case graph can be reduced to the problem of designing a revenue-optimal auction.  For example, the fact that on any day there is only one costly edge is a consequence of the fact that the optimal pricing menu for a single-parameter agent has a single deterministic option~\cite{Myerson}.  

We can leverage this description of the worst-case task graph to develop bounds on the procrastination ratio as a function of the present-bias distribution.  In optimal pricing theory, a special subclass of Pareto distributions, often referred to as {\it equal-revenue} distributions, are useful for worst-case examples.  For a given present-bias distribution, we calculate the ``smallest'' equal-revenue distribution that stochastically dominates it. As worst-case procrastination ratios grow with stochastic dominance, this gives us an upper bound on the worst-case procrastination ratio of any given present-bias distribution.  One can similarly derive lower bounds for any distribution that is \emph{not} dominated by a particular equal-revenue distribution.  One implication of our analysis is that present-bias distributions can be roughly divided into two categories:  those with light tails tend to have low (linear or even constant) worst-case procrastination ratios, whereas heavy-tailed distributions tend to have worst-case procrastination ratios that are exponential in $\nsteps$.

These worst-case examples lead us to ask whether there are natural conditions on the task graph, as well as the present-bias distribution, that lead to smaller procrastination ratios.  Intuitively, the key feature of the worst-case graph discussed above is that sub-optimal planning can cause an agent to not only incur extra costs today, but also reach a state from which the agent is strictly worse off; i.e., the cost of completing the task can be higher than it was initially.  We say that a graph has the {\em bounded distance property} if the weight of the shortest path from any node $v$ to the target $t$ is at most the weight of the shortest path from the initial node $s$ to $t$.  This condition captures scenarios, like training for a marathon, in which each day of training improves preparedness whereas procrastination diminishes it, but not below the initial base level. We show for any distribution $\dist$, the procrastination ratio of a bounded shortest-path graph is at most linear in $\nsteps$, and that this bound is tight even for distributions that have a constant probability of $\biasi=1$.



We next consider a stronger condition under which the procrastination ratio is constant.  We say a graph has the {\em monotone distance property} if, for any $s-t$ path $(s=v_0, v_1, \ldots, v_{k-1}, t=v_k)$, the shortest path from $v_i$ to $t$ is decreasing in $i$.  Roughly speaking, this condition captures scenarios in which progress made by an individual is not lost: regardless of an agent's action in the present round, the total cost required to complete the task has not increased.  For example, consider a skier deciding whether to rent or buy skies each weekend of the skiing season: as the cost of buying skies remains constant over time, the total cost required to ski the remainder of the season does not grow. We show that if the distribution over present-bias parameters has sufficient mass at biases close to $1$,\footnote{See Section~\ref{subsec:monotone} for a formal statement of the condition on the distribution.  For example, the condition is satisfied whenever there is a positive probability that the bias is equal to $1$.} then the procrastination ratio of graphs with the monotone distance property is bounded by a constant.  Moreover, the condition on the present-bias distribution is necessary for this result: if the agent has constant bias $\bias > 1$, there exist monotone-shortest-path graphs for which the procrastination ratio is $\Omega(\nsteps)$.\footnote{Akerloff's original example of mailing a package is such an example: it satisfies the monotone distance property and has linear procrastination ratio.  See Kleinberg and Oren~\cite{KleinbergO14} for a graph-theoretic formulation of that example.}

\subsection{Related work}

\paragraph{Behavioral Economics} Behavioral economics and game theory study models that can explain 
anomalies in human behavior which are consistently observed in experimental data, see for example ~\cite{Camerer03}. 
Our work adopts the model of~\cite{Akerlof} which aims to explain individual's procrastination that violates 
classical assumption of utility-maximizing individual's behavior. Closely related to procrastination are issues of 
abandonment of long-time projects~\cite{DonoghueRabin08} and benefit of imposing a deadline in the context of task 
completion~\cite{ArielyW02}. Other examples include models addressing attentiveness issues, e.g., where reduction 
of choice among the options available to an agent~\cite{DonoghueRabin99}, or delayed notification messages~\cite{JeffEly} 
may improve agent's performance.

Some behavioral models study various levels of agents' rationality. For example  Kaur, Kremer, and Mullainathan~\cite{self-control} study
sophisticated customers who are aware of their possible future procrastination, and 
Wright and Leyton-Brown~\cite{LeytonW10,LeytonW12,LeytonW14} model agents' behavior via quantal cognitive hierarchies in
single-round simultaneous-move games where agents are empirically observed not to follow Nash equilibrium strategies.

\paragraph{Hyperbolic Discounting} A significant amount of work in economics literature has been devoted to the study of {\em quasi-hyperbolic
discounting}, see~\cite{FrederickODonoghue} for a survey. This form of discounting function generalizes Akerlof's model by modeling 
agent's behavior with two parameters: present-bias factor $\bias$ and discounting factor $\delta\le 1$. In this model the agent at every step scales up the 
immediate cost by $\bias$, and discounts every cost $t$ steps away into the future by a factor of $\delta^t$.
This form of discounting has been used in many areas including addiction~\cite{HDiscount3}, self-control~\cite{HDiscount2}, 
and health and pre-retirement savings~\cite{HDiscount1}.
 
\paragraph{Graphical Model} Our work uses the model of procrastination proposed by Akerlof~\cite{Akerlof} and adapts the graphical 
framework of Kleinberg and Oren~\cite{KleinbergO14}.
The latter work studies graphical properties of the network which may cause high 
procrastination rates, and identifies a graph-minor condition on the task graph that implies a bound on the procrastination ratio. The follow-up work of Tang et al.~\cite{Tang14} 
provides improved bounds on the procrastination ratio, again as a function of this graph-minor condition.

\paragraph{Techniques} We model an individual as an agent performing graph traversal under uncertainty.  Related models also appear 
in the literature on route planning under uncertainty \cite{NikolovaK08}, or risk-aversion in routing and congestion 
games\cite{NikolovaM11,NikolovaM14,NikolovaS15} but under different objectives. 

Some of our technical results exploit a connection between agent behavior under hyperbolic discounting and auction theory.  In particular, we use tools characterizing the space of optimal single-item auctions under a distribution of agent values, as explored by Myerson~\cite{Myerson}.

\section{Model} 
\label{sec:model}

We consider a setting in which an agent must complete a task over a period of $\nsteps$ days.  We track the progress of the agent towards task completion using states of intermediate progress, which are nodes of a task graph $G=(\cup_{i=1}^{\nsteps+1} V_i,E)$ with edge weights $w:E\rightarrow\mathds{R}^+$.  
The set of nodes $V_i$ are the possible states of the agent on day $i$.  We assume without loss of generality that $V_1=\{s\}$, the {\it start} node, and $V_{n+1}=\{t\}$, the {\it end} node.  All edges occur between nodes on consecutive days, i.e., $v_iv_j\in E$ if and only if $v_i\in V_i$ and $v_j\in V_{i+1}$.  The weights represent the cost of transitioning between states: i.e., an agent in state $v_i$ wishing to transition to a state $v_{i+1}$ incurs a cost of $w(v_iv_{i+1})$.  The agent starts from state $s=v_1$ on day $1$, follows a path $(v_1v_2,\ldots,v_nv_{n+1})$ in $G$, and ends in state $t=v_{n+1}$ on day $\nsteps+1$, for a total cost of $\sum_{i=1}^nw(v_iv_{i+1})$.\footnote{Note we do not assume the agent completes the task; task abandonment can be represented by having a costly edge from a state $v_n\in V_n$ to $t$.}  For convenience we will refer to the shortest path between two nodes $v_i$ and $v_j$ in $G$ as the distance, denoted $d(v_i,v_j)$.  Thus the minimum cost way to complete the task is $d(s,t)$.

We are interested in the inefficiency of the traversed path due to the variable present-bias of the agent.  As such, for each day $i$ we introduce a present-bias factor $\biasi$ drawn i.i.d.\ from a given distribution $\dist$ supported on $[1,\infty)$.  When calculating costs on day $i$, the agent inflates costs incurred on that day by $\biasi$ while leaving remaining costs unweighted.  Thus an agent currently in state $v_i$ on day $i$ will choose a transition $v_iv_{i+1}$ minimizing $$\biasi\cdot w(v_iv_{i+1})+d(v_{i+1},t).$$
Note that on the following day $i+1$ the agent again chooses an edge using a (potentially different) present-bias factor $\biasi[i+1]$, and may therefore choose a state $v_{i+2}$ different from his planned path on day $i$.

We consider the ratio of the weight of the chosen path $v_1v_2,\dots,v_{\nsteps}v_{\nsteps+1}$ with the weight of the shortest $s-t$ path: \[
\frac{\sum_{i=1}^{\nsteps} w(v_iv_{i+1})}{\dtot}.
\]
The {\em procrastination ratio} is then the expectation (over the choice of the present bias factors) of this ratio.

In summary, an instance of our problem consists of a weighted task graph $G$ and present-bias distribution $\dist$.  We would like to understand the structure of worst-case instances, i.e., the distributions and tasks graphs that achieve the maximum procrastination ratio.  We would also like to describe natural classes of instances with small procrastination ratios.  To this end, it is useful to introduce the following restrictions on task graphs.

\begin{definition}[Bounded Distance]
\label{def:bounded_sp}
A weighted task graph $G$ satisfies \emph{bounded distance}, if $\forall i$, and $\forall v_i\in V_i$, $\dtot[v_i]\le\dtot$.
\end{definition}
Intuitively, the bounded distance property captures scenarios in which one's current state, in terms of the optimal cost to reach the goal, is never worse than the initial state.  For example, this includes scenarios in which ``starting over'' is always a free and feasible option: i.e., from each vertex there is a $0$-cost edge to a vertex whose distance from $t$ is at most $\dtot[v_i]$.

\begin{definition}[Monotone Distance]
\label{def:monotone_sp}
A weighted graph $G$ satisfies \emph{monotone distance} if 
$\forall i$ and $\forall v_iv_{i+1}\in E$ with $v_i\in V_i, v_{i+1}\in V_{i+1}$, $\dtot[v_i]\ge \dtot[v_{i+1}]$.
\end{definition}
This captures scenarios where the agent always makes progress towards the end state, i.e., no transition causes the agent to lose any of his prior accomplishments. We note that the latter restriction implies the former one, in the sense that any graph with monotone distances also has bounded distances.

\paragraph{Remark} Kleinberg and Oren~\cite{KleinbergO14} studied task planning in generic weighted directed acyclic graphs $G'=(V',E')$ with a designated start node $s$ and end node $t$, and presented the procrastination ratio as a function of the number of nodes in $G'$, say $N$.  To transform their setup to ours, simply add a self-loop to $t$ of weight $0$, create $N$ sets $V_i$ where $V_i$ contains copies of all nodes $v\in V'$ that can be reached from $s$ via a path of exactly $i-1$ edges in $G'$, and create weighted edges between nodes of $V_i$ and $V_{i+1}$ if such edges exist in $E'$.  Note that under this transformation, the number of steps $n$ is at most $N$, and the total number of vertices in the constructed graph is $O(N^2)$.

\section{Examples}
\label{sec:examples}

To illustrate the model and definitions introduced in Section~\ref{sec:model}, we work through  instantiations of the scenarios discussed in the introduction.  Throughout the course of this discussion, we will assume an agent with a present-bias factor $\biasi$ drawn from distribution $\dist$ which takes value $1$ with probability $1/3$ and value $3$ with probability $2/3$.  Note that this distribution is ``close to rational'' in the sense that the agent has a constant probability of $\biasi=1$.  In each example, we will normalize the shortest $s-t$ path to be $1$ and so the procrastination ratio is just the expected cost induced by the agent.

\subsection{Students Preparing Homework}
Our first scenario in the introduction concerned a student assigned homework, where the student's memory of the requisite material diminishes as time passes.  Suppose the homework is assigned on day $1$ and due on day $n+1$.  We construct a task graph $G=(\cup_{i=1}^{\nsteps+1}V_i,E)$ with $V_1=\{s_1\}$ and $V_{\nsteps+1}=\{t_{n+1}\}$.  On each day $i\in{2,\ldots,n}$, there are two possible states $s_i$ and $t_i$ corresponding to whether the student completed the assignment ($s_i$) or not ($t_i$).  The edges are: $E=\{s_is_{i+1} | i=1\ldots n-1\}\cup\{t_it_{i+1}| i=2\ldots n\}\cup\{s_it_{i+1}| i=1\ldots n\}$. The cost of transitioning from $s_i$ to $s_{i+1}$ or from $t_i$ to $t_{i+1}$ is zero.  The cost of transitioning from $s_i$ to $t_{i+1}$ is $w_{s_it_{i+1}}=2^{i-1}$.  This models the idea that the mental effort to complete the assignment grows by a factor of $2$ each day as the student's recollection of the material from lecture becomes dimmer.  See Figure~\ref{fig:ex1} for a representation of this task graph.  

\begin{figure}[t]
\centering
\subfigure[Training for a marathon.]{\label{fig:ex2}{\includegraphics[width=3in]{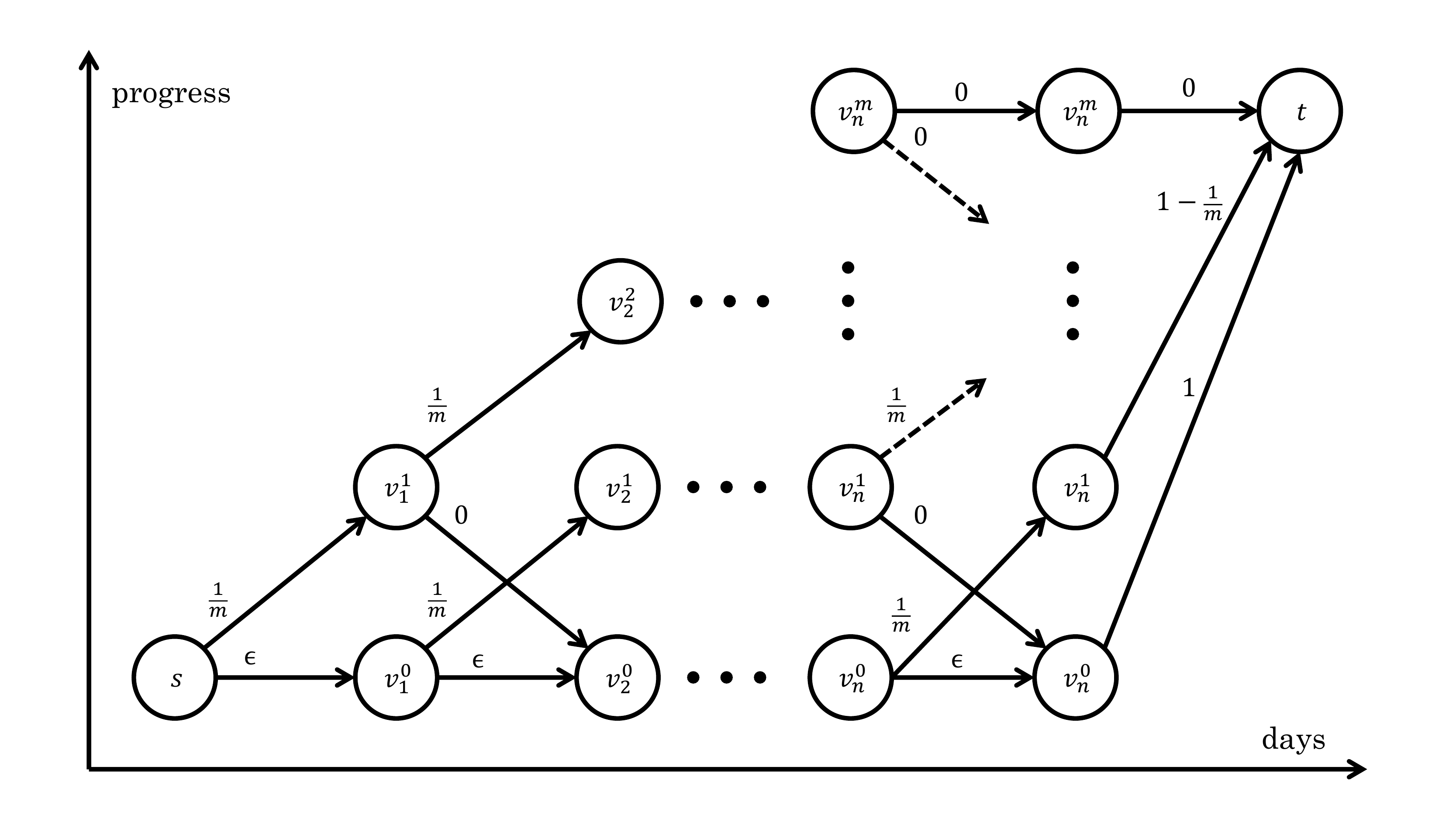}}}\\
\subfigure[Completing homework.]{\label{fig:ex1}{\includegraphics[width=2.5in]{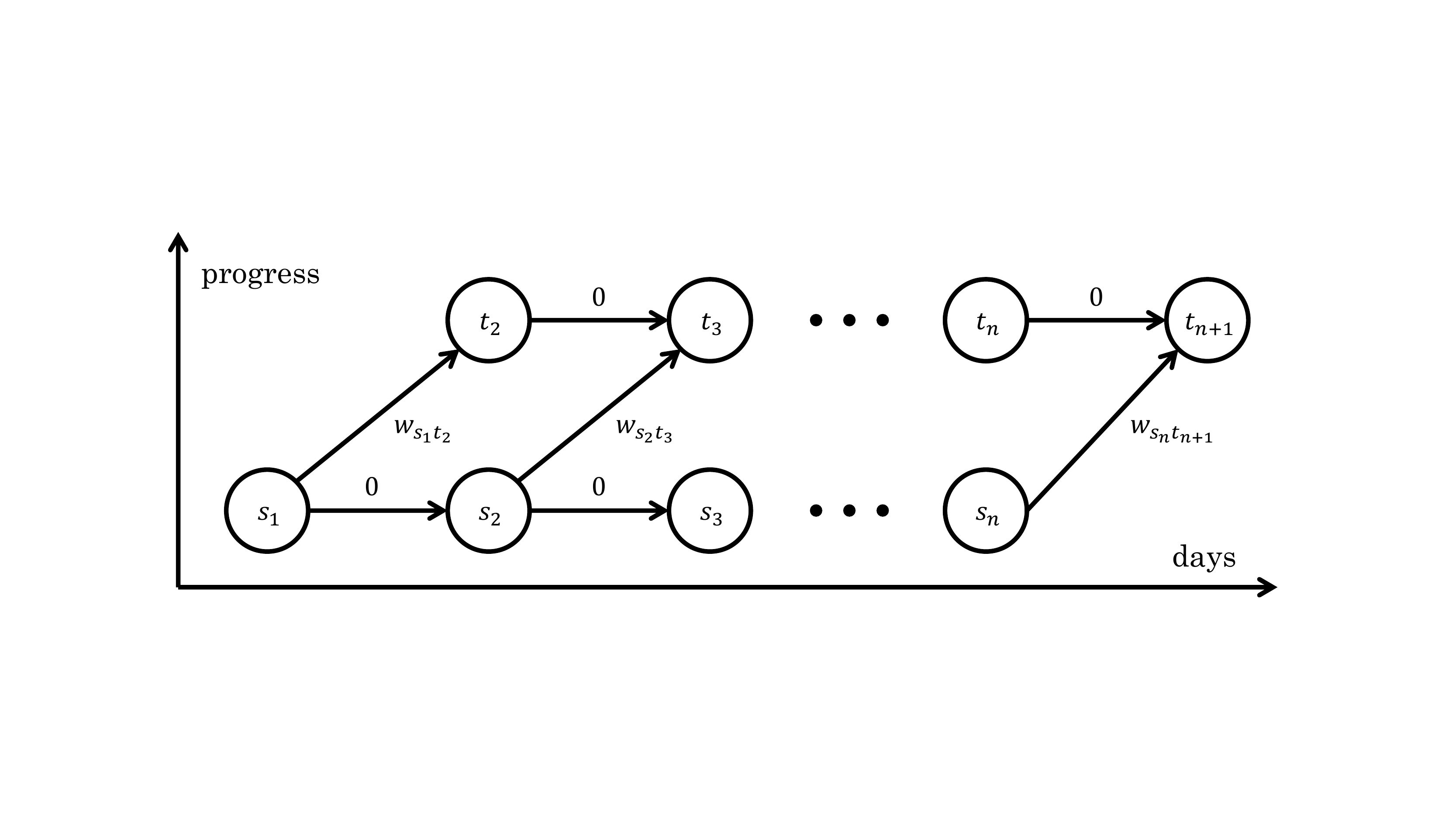}}}\quad
\subfigure[Renting skis.]{\label{fig:ex3}{\includegraphics[width=2.5in]{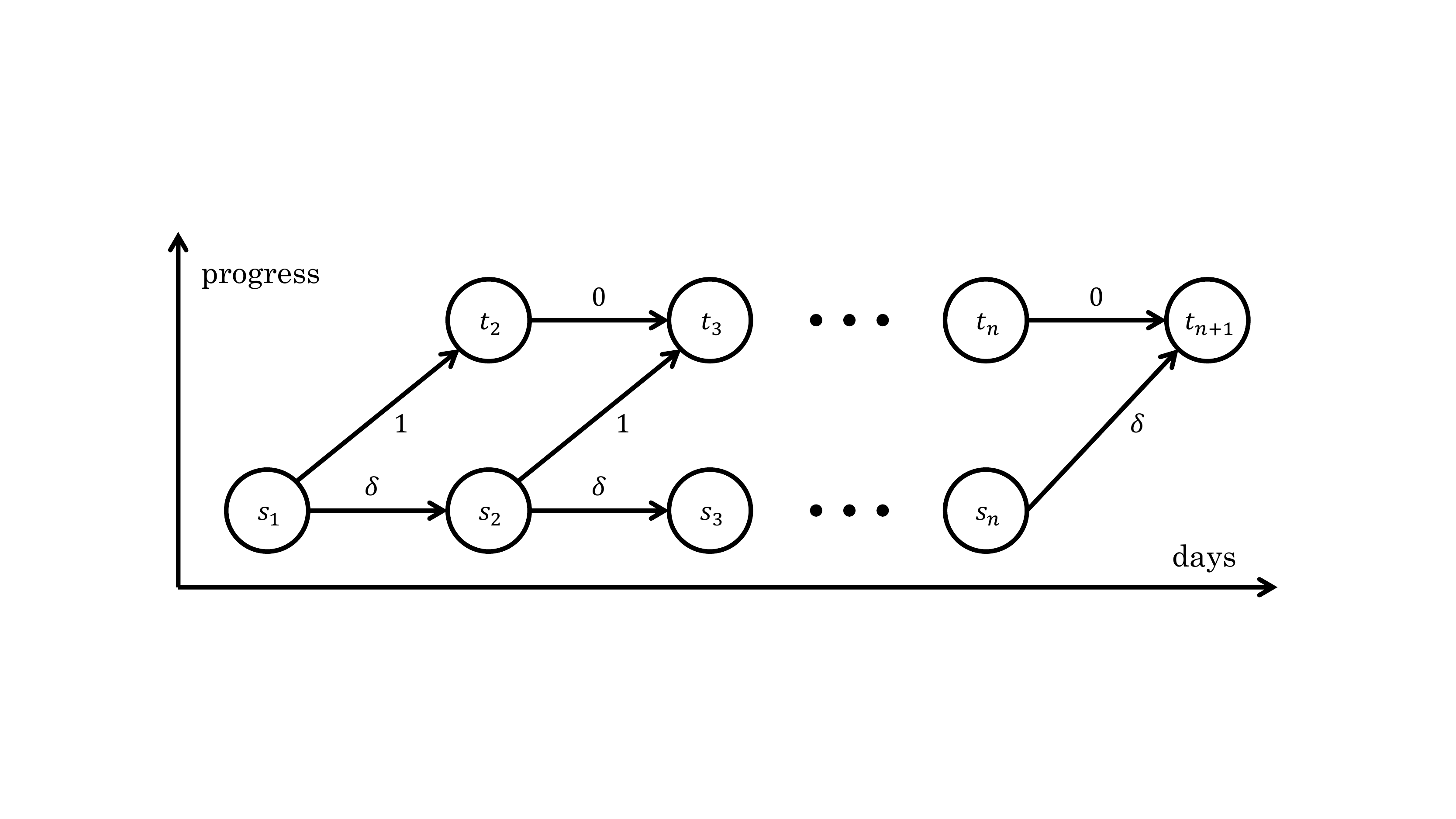}}}
\caption{Task completion graphs for the examples in Section \ref{sec:examples}.}
\label{fig:ex}
\end{figure}

Note that this task graph does not have the bounded distance property.  Accordingly, the procrastination ratio can be exponential, and indeed is exponential for the distribution $\dist$.  To see this, note on day $i$, the perceived cost of completing the assignment that day is $\biasi\cdot 2^{i-1}$ while the cost of waiting to complete the homework on the next day is $2^i$.  Thus, if the instantiation of $\biasi$ is greater than $2$, then the student will decide to complete the assignment on the following day.  For the present-bias distribution $\dist$ described above, this event occurs with probability $2/3$, and so the total expected cost incurred by the student is $\sum_{i=1}^{\nsteps+1} (2/3)^i2^{i-1}$, which is exponential in $\nsteps$.

\subsection{Runners Training for a Marathon}
Our second scenario in the introduction concerned a runner training for a marathon.  The runner's fitness diminishes each day she sleeps in instead of going for a training run.  Suppose the runner starts training on day $1$ and the marathon is scheduled for day $n+1$.  The task graph $G=(\cup_{i=1}^{\nsteps+1}V_i,E)$ has $V_1=\{s\}$ and $V_{\nsteps+1}=\{t\}$.  The runner has a fitness level in $\{0,\ldots,m\}$ where $0$ is the base-level fitness and $m$ is peak fitness.  If the runner is at fitness level $j<m$ on day $i<n$, then the runner can pay $\frac{1}{m}$ to train and increase fitness to level $j+1$ on day $i+1$.  Alternatively, if the runner is at fitness level $j>0$ on day $i<n$, then the runner can sleep in, paying $0$, and drop to fitness level $j-1$ on day $i+1$.  We assume the runner can maintain fitness level $m$ for $0$ cost and level $0$ at a small cost $\epsilon<\frac{2}{3m}$.\footnote{At fitness level $m$, running ceases to become a chore. At fitness level $0$, lack of exercise causes remorse and/or disutility from a low fitness level.}  On day $n$, if the runner is at fitness level $j$, then she pays $1-\frac{j}{m}$ to run the marathon.  See Figure~\ref{fig:ex2} for a representation of this task graph.

Note that this task graph has the bounded distance property.  As we show in Section~\ref{sec:bounded.distance}, such graphs have at most a linear procrastination ratio, for any distribution over present-bias parameters.  For $\biasi\sim\dist$, this task graph nearly achieves this bound.  To see this, note that on day $i<n$, if the runner is at fitness level $0<j<m$, then she perceives training as costing her $\biasi\cdot\frac{1}{m}+(1-\frac{j+1}{m})$ and not training as costing her $1-\frac{j-1}{m}$.  Therefore, she trains when $\biasi=1$ and sleeps in when $\biasi=3$ (ignoring boundary conditions at fitness level $0$).  This is a bounded random walk on the integers from $\{0,\ldots,m\}$ with an absorbing state at $m$.  The transition probabilities are biased toward lower fitness levels, and so the runner will spend most of her time transitioning between the low states, leading to a total cost of $\Theta(n/m)$ so long as $m$ is sufficiently large that she never hits the absorbing state.  Taking $m=\Omega(\log n)$ suffices, and this leads to a bound of $\Theta(n/\log(n))$.  In Section~\ref{sec:bounded.distance}, we show how to improve this example to get a linear bound for a broad class of present-bias distributions.

\subsection{Skiers Renting or Buying Skies}
Our third scenario in the introduction concerned a skier debating whether to rent or buy skies.  The skier buys a season pass on day $1$ and commits to skiing every weekend for the next $n$ weeks.  We construct a task graph $G=(\cup_{i=1}^{\nsteps+1}V_i,E)$ with $V_1=\{s\}$ and $V_{\nsteps+1}=\{t\}$.  On each day $i\in{2,\ldots,n}$, there are two possible states, {\bf rent} and {\bf own}, corresponding to whether the skier uses rental equipment or his own equipment.  Transitioning from a {\bf rent} state (or $s$) to an {\bf own} state on the following weekend costs $1$, the cost of buying skies.  Transitioning from a {\bf rent} state (or $s$) to another {\bf rent} state (or $t$) costs $\delta$, the cost of renting skies.  Once the skier is in the {\bf own} state, all future transitions are free.  See Figure~\ref{fig:ex3} for a representation of this task graph.

Note that this task graph has the monotone distance property.  As we show in Section~\ref{subsec:monotone}, if the present-bias is ``close enough'' to $1$ sufficiently often, such graphs have a constant procrastination ratio.  Indeed, for $\delta<2/3$, we can bound the procrastination ratio by noting that the skier will choose to buy when $\biasi=1$ (except possibly towards the very end of the season) and rents otherwise.  The cost is then at most $\sum_{i=1}^n(\frac{2}{3})^i(\delta\cdot i+1)$, which is a constant.  Note that if the present-bias was bounded away from $1$, e.g., a point mass at $3$, then the skier would rent for the entire season and pay a linear cost.

\section{Worst-Case Procrastination}
\label{sec:worst-case}

For a fixed present-bias distribution, we would like to characterize the worst-case procrastination ratio over all task graphs with $n$ days.  To this end, we first derive the form of a task graph that achieves the worst-case procrastination ratio for an arbitrary fixed present-bias distribution.  We then develop a natural parameterization of distributions and bound this worst-case ratio as a function of the distribution parameters.

\subsection{Worst-Case Task Graphs}
\label{sec:worst.graph}

Given present-bias distribution $\dist$, we will show that there is a worst-case task graph with a simple form reminiscent of that from the student/homework example.  On each day there are only two possible states: the uncompleted state and the completed state.  Moving between uncompleted states or completed states has zero cost while moving from an uncompleted state to a completed one has a cost which is a function of the distribution.  

The proof uses the following theorem from {\it optimal pricing theory}: A buyer has a value for an item, drawn from a distribution, and faces a menu of options, each consisting of a probability of allocation and a price.\footnote{The menu must contain the option of not buying for a cost of zero, i.e., the buyer is allowed to walk away.} The buyer, acting to maximize his expected utility (defined as value times allocation probability minus price), induces an expected allocation and price of the menu.  The theorem states that for any linear function of expected allocation and price, there is a simple {\it single posted price} menu that optimizes that function.
\begin{theorem}[~\cite{Myerson}]
\label{thm:optpricing}
Let $M\subseteq[0,1]\times \mathds{R}$ be such that $(0,0)\in M$.  For any such $M$ and $v\in\mathds{R}^+$, write $(x_M(v),p_M(v))=\argmax_{(x,p)\in M}[vx-p]$.  Then for all distributions $\dist$ with support contained in $\mathds{R}^+$, and for all $\alpha,\beta\in\mathds{R}$, there exists a $p^*$ and corresponding $M^*=\{(0,0),(1,p^*)\}$ such that
$$\alpha E_{v\sim\dist}[x_M(v)]+\beta E_{v\sim\dist}[p_M(v)]\leq\alpha E_{v\sim\dist}[x_{M^*}(v)]+\beta E_{v\sim\dist}[p_{M^*}(v)].$$
\end{theorem}
We will construct the worst-case task graph recursively.  From a node $v_i$, the subproblem requires us to construct neighbors $v_{i+1}$, edge weights $w(v_iv_{i+1})$, and the weights of the shortest paths $d(v_{i+1}t_{n+1})$ from these neighbors to the end state, all subject to the constraint that the weight of the shortest path $d(v_i,t_{n+1})$ from $v_i$ to the end state is preserved.  We solve this subproblem by relating the ratio $w(v_iv_{i+1})/d(v_i,t_{n+1})$ to the allocation probability and $d(v_{i+1},t_{n+1})/d(v_i,t_{n+1})$ to the price of a menu in pricing theory.  

\begin{theorem}
\label{th:optimal_contract} Given a present-bias distribution $\dist$, there is a task graph over $n$ days, $G_\dist^n=(\cup_{i=1}^{n+1}V_i,E)$, achieving maximum procrastination ratio with the following form.  Let $V_1=\{s_1\}$, $V_{n+1}=\{t_{n+1}\}$, and $d(s_1,t_{n+1})=1$.  Then
\begin{itemize}
\item for $i\in{2,\ldots,n}$, $V_i=\{s_i,t_i\}$,  
\item for $i\in{1,\ldots,n-1}$, $w(s_is_{i+1})=0$,
\item for $i\in{2,\ldots,n}$, $w(t_it_{i+1})=0$,
\item and for $i\in{1,\ldots,n}$, $w(s_it_{i+1})=d(s_i,t_{n+1})\leq d(s_{i+1},t_{n+1}),$
\end{itemize}
where $d(s_i,t_{n+1})$ is chosen as a function of $\dist$.
\end{theorem}

\begin{proof}
First note the assumption $d(s_1,t_{n+1})=1$ is without loss of generality as we can always normalize all weights by the initial shortest path without changing the procrastination ratio. We prove the theorem by constructing a worst-case task graph recursively.  In the $i$'th iteration, for $i=1\ldots n-1$, we will construct:
\begin{itemize}
\item the nodes of set $V_{i+1}$, 
\item the weights of edges from $V_i$ to $V_{i+1}$, 
\item and the shortest continuation path weights $d(v_{i+1},t_{n+1})$ for $v_{i+1}\in V_{i+1}$, 
\end{itemize} 
to maximize the expected procrastination ratio of the chosen path subject to constraints on the shortest paths $d(v_i,t_{n+1})$ for $v_i\in V_i$.  

We first describe the behavior of the agent at node $v_i$ for a fixed set of transition options $V_{i+1}$ consisting of weighted edges $w(v_iv_{i+1})$ and continuation distances $d(v_{i+1},t_{n+1})$ for $v_{i+1}\in V_{i+1}$. An agent with present-bias parameter $\biasi$ selects transition $v^*_{i+1}$ minimizing $\biasi w(v_iv_{i+1})+d(v_{i+1},t_{n+1}).$  This minimization problem is equivalent to the following maximization problem, as can be seen by first changing the sign of the objective, then adding the constant $\biasi\cdot d(v_i,t_{n+1})$, and then scaling by the constant $1/d(v_i,t_{n+1})$:
$$\max_{v_{i+1}\in V_{i+1}}\left(b_i\left(1-\frac{w(v_iv_{i+1})}{d(v_i,t_{n+1})}\right)-\frac{d(v_{i+1},t_{n+1})}{d(v_i,t_{n+1})}\right).$$
The above problem is the optimization problem faced by a buyer in a pricing menu where: the value of the buyer is $\biasi$, and the allocation and price of menu item $v_{i+1}$ are $x(v_{i+1})=1-\frac{w(v_iv_{i+1})}{d(v_i,t_{n+1})}$ and $p(v_{i+1})=\frac{d(v_{i+1},t_{n+1})}{d(v_i,t_{n+1})}$, respectively.  We argue this menu satisfies the conditions of Theorem~\ref{thm:optpricing}.  As $d(v_i,t_{n+1})$ is the weight of the shortest $v_i-t_{n+1}$ path, there must be an option $v_{i+1}$ with $w(v_iv_{i+1})\leq d(v_i,t_{n+1})$.  Furthermore, any option $v'_{i+1}$ with  $w(v_iv'_{i+1})> d(v_i,t_{n+1})$ will never be preferred to $v_{i+1}$ and so we can eliminate such nodes from the graph.  Therefore, the allocations $x(v_{i+1})$ are in $[0,1]$.  Furthermore, we can assume without loss of generality that there is a transition $t_{i+1}\in V_{i+1}$ with $w(v_it_{i+1})=d(v_i,t_{n+1})$ and $d(t_{i+1},t_{n+1})=0$ as we are required to include a transition option with $w(v_iv_{i+1})+d(v_{i+1},t_{n+1})=d(v_i,t_{n+1})$ and any such option weakly dominates $t_{i+1}$ for any bias factor.  
This proves that the menu contains an option $t_{i+1}$ with $x(t_{i+1})=0$ and $p(t_{i+1})=0$.

To apply Theorem~\ref{thm:optpricing}, we must argue that the objective of maximizing the procrastination ratio is linear in the expected allocation and pricing of any menu $V_{i+1}$ defined above. Let $\opt(i,d)$ denote the expected cost of the agent in the worst-case task graph with $n-i$ days and shortest path weight $d$.  We can write $\opt(i,d(v_i,t_{n+1}))$, the worst-case expected cost conditioned on being at $v_i$, recursively as:
$$\max_{V_{i+1}}\Big[\Ex[\biasi\sim\dist]{w(v_iv^*_{i+1})}+\Ex[\biasi\sim\dist]{\opt(i+1,d(v^*_{i+1},t_{n+1}))}\Big],$$
where $v^*_{i+1}$ is the choice of the agent with bias $\biasi$ facing options $V_{i+1}$, and the base of the recursion $\opt(n,d)=d$.  Using the linearity of $\opt(i,d)$ in its second argument and the menu pricing notation introduced in the preceding paragraph, we see the optimization problem is equivalent to:
\be
\label{eq:designer_objective}
\max_{V_{i+1}}\Big(1-E_{\biasi\sim\dist}[x(v^*_{i+1})]+E_{\biasi\sim\dist}[p(v^*_{i+1})]\opt(i+1,1)\Big)\cdot d(v_i,t_{n+1}).
\ee
Applying Theorem~\ref{thm:optpricing}, we conclude that there is an optimal menu $V_{i+1}=\{s_{i+1},t_{i+1}\}$ where $x(s_{i+1})=1$, $p(s_{i+1})=p^*\geq1$ (as the value $\biasi\geq1$), $x(t_{i+1})=0$, and $p(t_{i+1})=0$.  Given a partially constructed task graph with $V_i=\{s_i,t_i\}$, we can therefore optimally extend it by defining $V_{i+1}=\{s_{i+1},t_{i+1}\}$.  The above discussion shows that the transitions from $s_i$ to $V_{i+1}$ should have edge weights $w(s_is_{i+1})=0$ and $w(s_it_{i+1})=d(s_i,t_{n+1})$ and continuation shortest path weights $d(s_{i+1},t_{n+1})\geq d(s_i,t_{n+1})$ and $d(t_{i+1},t_{n+1})=0$.  To satisfy the constraint that $d(t_i,t_{n+1})=0$ we also add an edge from $t_i$ to $t_{i+1}$ of weight zero.  This completes the $i$'th iteration.  To complete the construction, add edges $s_nt_{n+1}$ and $t_nt_{n+1}$ with weights $d(s_n,t_{n+1})$ and $0$, respectively.
\end{proof}

\subsection{Bounding the Procrastination Ratio for a Bias Distribution}
\label{sec:distrib.ratio}

In Section~\ref{sec:worst.graph} we characterized the task graph that maximizes the procrastination ratio for any given present-bias distribution.  In this section we leverage this characterization to explore how properties of the present-bias distribution impacts the worst-case procrastination ratio.  In particular, we are interested in determining features of the bias distribution that imply sub-exponential, or even constant, bounds on the procrastination ratio.

We first note that worst-case procrastination ratio can only be greater for distributions that generate higher present-bias parameters.  
Recall that distribution $\distb$ {\it stochastically dominates} distribution $\dist$ if, for all $x$, $\Pr_{b\sim\distb}[b>x]\geq\Pr_{b\sim\dist}[b>x]$.  That is, $\distb(x)\leq\dist(x)$ for all $x$.

\begin{lemma}
\label{th:rev_monotone}
Given two distributions $\dist$ and $\distb$ such that $\distb$ stochastically dominates $\dist$, the worst-case procrastination ratio for $\distb$ is at least as large as the worst-case procrastination ratio for $\dist$.
\end{lemma}

\begin{proof}
Fix $n$ and let $G_\dist^n$ be a worst-case task graph for $\dist$ of the form specified in Theorem~\ref{th:optimal_contract}. We use $\revi$ to denote the expected cost paid by an agent with $\biasi\sim\dist$ on days $i,\dots, \nsteps$ in the graph $G_\dist^n$, conditional on being at node $s_i$ on day $i$. Similarly, we use $\revbi$ to denote the expected cost paid by an agent with $\biasi\sim\distb$ on days $i,\dots, \nsteps$ in the graph $G_\dist^n$, conditional on being at node $s_i$ on day $i$. In the following we show using backward induction on $i$ that $\revbi\ge\revi$, implying that $\revbi[0]$, a lower bound on the worst-case procrastination ratio for $\distb$, is at least as large as $\revi[0]$, the worst-case procrastination ratio for $\dist$.

When $i=\nsteps$ we have $\revbi[\nsteps]=\revi[\nsteps]=d(s_{\nsteps},t_{\nsteps+1})$. For $i<n$ note that by definition of $G_\dist^n$, an agent with threshold bias $\thresh = d(s_{i+1},t_{\nsteps+1})/ d(s_i,t_{\nsteps+1})$ is indifferent between transitioning to $s_{i+1}$ at a cost of zero, and transitioning to $t_{i+1}$ at a cost of $d(s_i,t_{n+1})$ with zero future costs.  Thus
\begin{equation}
\begin{cases}
\revi & = (1-\dist(\thresh))\cdot  \revi[i+1] + \dist(\thresh)\cdot d(s_i,t_{n+1}) \\
\revbi & =  (1-\distb(\thresh))\cdot  \revbi[i+1] + \distb(\thresh)\cdot d(s_i,t_{n+1}).
\end{cases}
\end{equation}
By induction, $\revbi[i+1]\ge\revi[i+1]$, and so
\[
\revbi-\revi \ge \left(\dist(\thresh)-\distb(\thresh)\right)\cdot \left(\revi[i+1]-d(s_i,t_{n+1})\right).
\]
By stochastic dominance, $\dist(\thresh)\geq\distb(\thresh)$, and by the construction of $G_\dist^n$, $\revi[i+1]\ge d(s_{i+1},t_{n+1})\ge d(s_i,t_{n+1})$. As a result $\revbi-\revi\geq 0$, which completes the induction.
\end{proof}


Lemma~\ref{th:rev_monotone} motivates us to categorize distributions into classes, where all distributions in a certain class stochastically dominate a family of distributions.  Such a categorization, along with Lemma~\ref{th:rev_monotone}, will allow us to derive bounds on the procrastination ratio of any given distribution.  
%
%
Our bound will be with respect to the following parametrization of present-bias distributions, which is essentially the largest value $z$ such that the distribution stochastically dominates the equal-revenue distribution with revenue $z$.\footnote{The equal-revenue distribution with revenue $z$ satisfies $F(x) = 1 - \tfrac{z}{x}$ for all $x \geq z$.}


\begin{definition}
Given a distribution $\dist$ with support contained in $[1,\infty)$, write $z(\dist) \equiv \max{(1-\dist(\bias))\bias}$.  Given $z \in (0,\infty)$, we will write $\fdist(z)$ for the family of distributions $\dist$ with $z(\dist) = z$.  
\end{definition}

Note that any distribution $\dist$ with bounded support belongs to some family $\fdist(z)$, potentially with $z = \infty$.  For $\dist \in \fdist(\infty)$, i.e., $\dist$ has infinite support and $\dist\notin\fdist(z)$ 
for any $z>0$, then there are task graphs for $\dist$ with procrastination ratio growing faster than any exponential function in $\nsteps$.  We will largely focus on distributions for which $z(\dist)$ is finite.

To build some intuition for the
classes $\fdist(z)$, we give below a few examples.
\begin{enumerate}
\item The uniform distribution over the interval $[1,3]$ belongs to the family $\fdist(1.125)$, with $(1-\dist(1.5))\cdot 1.5=1.125$. 
\item The uniform distribution over the interval $[1,2]$ belongs to the family $\fdist(1)$. 
\item Suppose $\dist$ is the upper half of a normal distribution with mean $1$.  That is, $\bias=\max(\xi,1)$ where $\xi\sim N(1,1)$ is a normal random variable with mean and normal deviation equal to $1$.  In this case, we can maximize $\bias(1-\dist(\bias))$ numerically to 
find that $\dist(\bias)\in\dist(z)$ for $z\approx 0.507$.
\item Suppose $\dist(x)=1-\frac{1}{2\sqrt{x}}$ for $x\in[1,100)$, and $\dist(100)=1$.  This is a heavy-tailed distribution, which has an atom at $\bias = 1$ with probability $\frac{1}{2}$.  This distribution belongs to $\fdist(5)$.  
\end{enumerate}

We are now ready to bound the worst-case procrastination ratio for any given distribution $\dist$, as a function of $z(\dist)$.  The following two theorems tell us whether a given distribution has an exponential, linear, or constant procrastination ratio, in the worst case over task graphs.  Theorem~\ref{th:dist_family} provides an upper bound on the procrastination ratio, which is tight for the family $\fdist(z)$.  Theorem~\ref{th:dist_twopoints} provides a lower bound for any given distribution.  

\begin{theorem}
\label{th:dist_family}
Let $\dist\in \fdist(z)$ for some $z>0$. Then the procrastination ratio of any task graph for $\dist$ is at most $\sum_{i=0}^{k-1}z^i$ and this bound is tight.
\end{theorem}

\begin{proof}
To prove an upper bound on the procrastination ratio we analyze the following distribution $\dist$ with bounded support on $[1,C]$, where $C$ is some sufficiently large constant:
\[
\dist:
\left\{
\begin{array}{c l l}
\Pr[\bias=1]  & = 1 - z, & \text{ if } z<1\\
\Pr[\bias\le x]  & = 1 - z/x, & \text{ }x\in[\max(1,z),C] \\
\Pr[\bias=C]  & = z/C &
\end{array}\right .
\]
We note that $\dist$ stochastically dominates every distribution in $\fdist(z)$ with support bounded by $C$.


Recall from Theorem~\ref{th:optimal_contract} the form of the worst-case task graph for distribution $\dist$. We briefly describe the form of the graph here.  Given a current state $v_i$, the agent is offered two options: an edge with weight $\dtot[v_i]$ (and no further costs thereafter), and an edge with weight $0$ that transitions to a vertex $v_{i+1}$ with increased distance to $t$.  We claim that, in fact, $\dtot[v_{i+1}] = C\cdot \dtot[v_i]$ in the worst-case graph. 
To see why, note that according to \eqref{eq:designer_objective} the worst-case task graph maximizes a linear function
of $\Ex[\biasi\sim\dist]{x(v^*_{i+1})}$ and $\Ex[\biasi\sim\dist]{p(v^*_{i+1})}$. The former term does not depend on the
set of offered price $p$, because $\Ex[\biasi\sim\dist]{p(v^*_{i+1})}$ corresponds in \eqref{eq:designer_objective} to the expected revenue of a single item auction 
for the equal revenue distribution $\dist$. Therefore, $\Ex[\biasi\sim\dist]{p(v^*_{i+1})}=z$ for any menu of options, and hence the optimization problem reduces to maximizing the term $1-E_{\biasi\sim\dist}[x(v^*_{i+1})]$. It is maximized when the agent's probability of paying a positive cost at step $i$ is
maximized. This is maximized when agent can either increase the distance to $t$ to $\dtot[v_{i+1}]=C\cdot \dtot[v_i]$, or pay the total amount $\dtot[v_i]$ right away. 
Thus we obtain the worst-case task graph.

Now we calculate the procrastination ratio of this task graph. At step $i$, from a given vertex $v \in \{v_i, u_i\}$, the expected distance to $t$ is $\Ex{\dtot[v]}=z^{i-1}$. The probability that agent takes a non-zero edge is $1-1/C$. Thus at the steps $i=1,\dots,\nsteps-1$ the agents pays a cost, in expectation, of $\dist(C)\cdot \Ex{\dtot[v]} = (1-1/C) \cdot z^{i-1}$ and $\Ex{\dtot[v]}=z^{\nsteps-1}$ at the last step. As a result, the total expected cost paid by the agent is

$$
\sum_{i=1}^{\nsteps-1}  (1-1/C) \cdot z^{i-1} + z^{\nsteps-1}
$$

As $C$ tends to infinity this quantity converges to $\sum_{i=1}^{\nsteps}  z^{i-1} =\frac{1-z^\nsteps}{1-z}$.  This is an upper
bound on the procrastination ratio for $\dist\in\fdist(z)$ by Lemma~\ref{th:rev_monotone}.  Furthermore, this bound is tight for the class $\fdist(z)$, since we have exhibited a distribution and a task graph that achieve this bound in the limit as $C \to \infty$. \qed
\end{proof}

\begin{theorem}
\label{th:dist_twopoints}
For arbitrary $\dist$, if $\left(1-\dist(\biasi[0])\right)\biasi[0]=z$ for some point $\biasi[0]>1$, then 
there is a task graph with procrastination ratio $\sum_{i=1}^{\nsteps-1}  (1-1/\biasi[0]) \cdot z^{i-1} + z^{\nsteps-1}$.
\end{theorem}

\begin{proof}
Indeed, we may assume that $\dist$ has finite support $\{1,\biasi[0]\}$, since every other distribution stochastically dominates this one
and by revenue monotonicity (Lemma~\ref{th:rev_monotone}) the adversary may only get higher procrastination ratio. 
The task graph from the proof of Theorem~\ref{th:dist_family} gives us the required bound. 
\end{proof}

With these results in hand, we can bound the worst-case competitive ratios for the examples described above:
\begin{enumerate}
\item The uniform distribution over the interval $[1,3]$ is in $\fdist(1.125)$, so Theorem~\ref{th:dist_family} implies that the worst-case procrastination ratio of any task graph for this distribution is not more than $O(1.125^\nsteps)$.  Moreover, Theorem~\ref{th:dist_twopoints} applied at point $b_0 = 1.5$ implies that there are task graphs with procrastination ratio of $\Omega(1.125^\nsteps)$.
\item The uniform distribution over the interval $[1,2]$ is in $\fdist(1)$, and therefore has at most a linear procrastination ratio for any task graph (Theorem~\ref{th:dist_family}).  On the other hand, one can obtain sublinear bound of $\Omega(\sqrt{\nsteps})$ in Theorem~\ref{th:dist_twopoints} by taking $\bias_0=1+\eps$ and $z=1-\eps^2$ for any $\eps > 0$.
\item For the upper half of a normal distribution with mean $1$, since the distribution lies in $\fdist(z)$ with $z \approx 0.507$, Theorem~\ref{th:dist_family} implies that the procrastination ratio is at most $2.03$.
\item For the distribution given by $\dist(x)=1-\frac{1}{2\sqrt{x}}$ for $x\in[1,100)$, and $\dist(100)=1$, since the distribution belongs to $\fdist(5)$ we have that the procrastination ratio is bounded by $O(5^\nsteps)$, and invoking Theorem~\ref{th:dist_twopoints} with $b_0 = 100$ shows that it can be as large as $\Omega(5^\nsteps)$.
\end{enumerate}

\section{Special Task Graphs}
\label{sec:special}

In this section we consider restrictions on the structure of a task graph.  We first study the impact of bounding the length of the shortest path to the target node, from any vertex in the network.  We then consider a stronger monotonicity property, which roughly states that following an edge in the graph can never increase the length of the shortest path to the target.  We will show that the former property implies a linear procrastination ratio, and this is tight, whereas the latter property implies a constant procrastination ratio for appropriate present-bias distributions.

\subsection{Bounded Distance}
\label{sec:bounded.distance}
Here we consider the class of task graphs such that distance from any given state to $t$ is bounded by the initial distance $\dtot$.  As it turns out, any such graph has at most a linear procrastination ratio, for any distribution $\dist$.

\begin{claim}\label{cl:upb-bsp}
The procrastination ratio of any protocol with bounded distance does not exceed $\nsteps$.
\end{claim}
\begin{proof}
Suppose the chosen path is $v_1, v_2, \dotsc, v_n, t$.  By assumption, $\disti{t} \leq \dtot$ for all $i$.  Since the agent chooses the edge $v_i v_{i+1}$ to minimize $b_i \cdot w(v_i v_{i+1}) + \disti[i+1]{t}$ where $b_i \geq 1$, and since the first edge of the path realizing distance $\disti{t}$ is itself an option, it must be that $w(v_i v_{i+1}) \leq \disti{t} \leq \dtot$ for each $i$.  The total cost of the chosen path is therefore at most $\nsteps\cdot\dtot.$ \qed
\end{proof}

The bound in Claim~\ref{cl:upb-bsp} holds for arbitrary procrastination ratios.  One might hope that if the present-bias distribution $\dist$ is sufficiently well-behaved, the procrastination ratio would improve.  However, as we now show, this bound on the procrastination ratio is asymptotically tight for any $\dist \in \fdist(z)$, for any $z > 1$.  In other words, one cannot hope to avoid a linear lower bound unless $\dist \in \fdist(z)$ for some $z \leq 1$, which we view as a particularly strict condition.

The rough idea for the construction is to simulate a random walk.  Whenever the present-bias parameter is low the agent will incur a cost and reduce the length of the shortest path to the goal, but whenever the present-bias parameter is high the length of the shortest path will increase.  Over the course of $\nsteps$ steps, the agent will have repeatedly made and lost progress sufficiently often to have accrued a total cost that is linear in $\nsteps$, with high probability.

\begin{theorem}
\label{thm:boundedsp_linear_lb}
If $\dist\in \fdist(z)$ for some $z>1$, then there is a task graph with bounded shortest path that has procrastination ratio $\Omega(\nsteps)$.
\end{theorem}

\begin{proof}
Let $\delta \in (0,1)$ and integers $\alpha, \beta \geq 1$ be constants that depend on $\dist$, to be determined later.   We construct a task graph as follows.  Each layer $V_i$ contains $\nsteps\beta$ vertices; say $V_i = \{v_{i,1},\dots,v_{i,\nsteps\beta}\}$.  Each vertex $v_{i,j}$ has out-degree at most $2$.  
For $i<\nsteps$, starting from vertex $v_{i,j}$,
\begin{itemize}
\item if $\alpha \leq j  \leq (\nsteps-1)\beta$, then the agent may go either to $v_{i+1,j-\alpha}$  at no cost, or to $v_{i+1,j+\beta}$ at a cost of $\delta^j- \delta^{j+\beta}$;
\item if $j < \alpha$, then the agent must transition to $v_{i+1,\alpha}$ at a cost of $\delta^j-\delta^{\alpha}$.
\item Otherwise, if $j > (\nsteps-1)\beta$, the agent must transition to $v_{i+1,j}$ at  no cost.
\end{itemize}
For $i=\nsteps$, each vertex $v_{i,j}$ has only a single outgoing edge $(v_{i,j}, t)$, which has cost $\delta^j$.

We claim that, for each vertex $v_{i,j}$, we have $\dtot[v_{i,j}] = \delta^j$.  This is achieved by the path in which the agent always chooses the option with positive cost, whenever multiple options are available.  That is, the agent repeatedly chooses to take the route that increases the second coordinate of its vertex, if possible.  If this path is $v_{i,j}, v_{i_2, j_2}, \dotsc, v_{i_\ell, j_\ell}, t$, then the total cost is 
\[ (\delta^j - \delta^{j_2}) + (\delta^{j_2} - \delta^{j_3}) + \dotsc + (\delta^{j_{\ell-1}} - \delta^{j_\ell}) + \delta^{j_\ell} = \delta^j,\] 
via a telescoping sum.  
The agent starts at $v_{1,0}$, and $\dtot[v_{1,0}]=1$.  We note that the task graph satisfies the bounded distance condition.



\begin{figure}
\begin{center}
\includegraphics[scale=1]{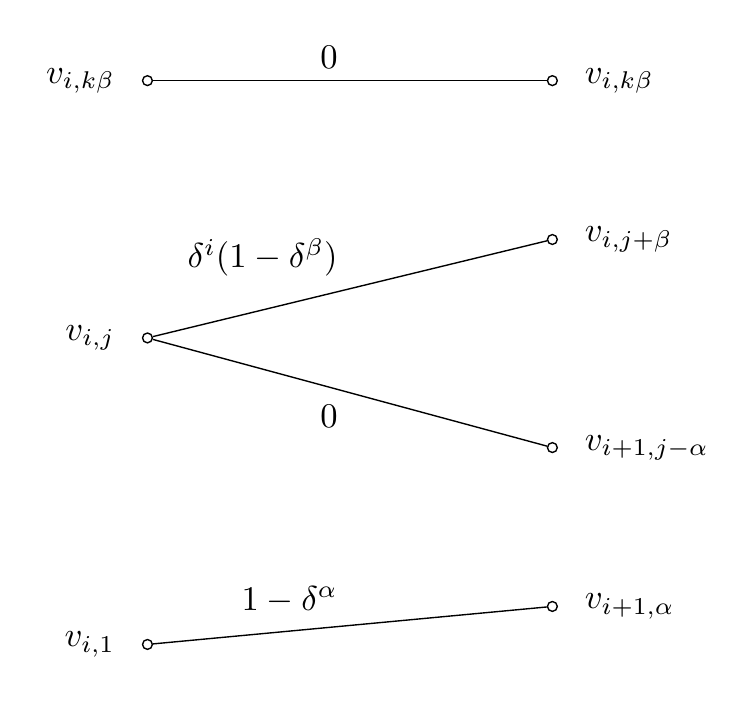}
\caption{Task completion graph for the protocol}
\label{fig:student}
\end{center}
\end{figure}

Let $\biasz$ be such that $(1-\dist(\biasz))\biasz= z > 1$; we know such a $\biasz$ exists by assumption.  We will then choose $\alpha$ and $\beta$ so that 
\[
\frac{1}{\biasz}<\frac{\beta}{\alpha+\beta}<1-\dist(\biasz).
\]
Note that the agent will select the $0$-cost option, if multiple options are available, when 
\[ \delta^{j-\alpha} < \biasi(\delta^j - \delta^{j+\beta}) + \delta^{j+\beta}.\] 
We claim that this will occur for $\biasi \geq \biasz$, if $\delta$ is sufficiently close to $1$.  To see this, define $H(\delta) \eqdef \biasz \cdot \delta^\alpha - \delta^{\alpha+\beta}\cdot (\biasz-1).$  Note that $H(1)=1$.
Furthermore, $H'(1)=\biasz\cdot \alpha - (\alpha+\beta) \cdot ( \biasz-1)= \alpha+\beta - \beta\cdot\biasz$ is negative if
$\frac{1}{\biasz} < \frac{\beta}{\alpha+\beta}$. The claim follows, as $H(\delta)= H(1) - H'(1)(1-\delta)+o\left(1-\delta\right)$ for
$\delta$ sufficiently close to $1$.  Choose $\delta$ to be sufficiently close to $1$ for the claim to hold.

%

Given $i < \nsteps$ and $j$ such that $\alpha \leq j \leq (\nsteps - 1)\beta$, let $\gamma$ be the expectation, over the present-bias distribution, of the change in index of the current vertex when an agent makes a single move starting at vertex $v_{i,j}$.  Then $\gamma$ is at most $\beta\dist(\biasz)-\alpha(1-\dist(\biasz))$, which is negative by our choice of $\alpha$ and $\beta$.

%
%
%

%
%
Note that since the agent started at $v_{1,0}$ it is impossible to reach a state $v_{i,j}$ where $j\geq (\nsteps-1)\beta$.
We now count the sum of all changes in vertex index, over the course of the agent's selected path, enumerated separately for two cases: (i) when $\alpha\leq j\leq (\nsteps-1)\beta$, and when (ii) $j<\alpha$.

Let random variables $s_{1}$ and $s_{2}$ denote the sum of changes in vertex index for these two cases. We observe that the total sum must be non-negative, and hence
\begin{equation}\label{boundeq1}
E[s_1]+E[s_2]\ge 0.
\end{equation}
%
Also, by our construction, if $j<\alpha$ then the subsequent step will have $j \geq \alpha$, and hence the second case applies in at least half of the rounds. Therefore,

\begin{equation}\label{boundeq2}
E[s_1] \le \gamma\cdot \nsteps/2
\end{equation}
where recall that $\gamma < 0$.
Equations~\eqref{boundeq1} and~\eqref{boundeq2} imply that
\begin{equation}\label{boundeq3}
E[s_2] \ge -\gamma\cdot \nsteps/2.
\end{equation}
Note that for each index-increasing step, the agent is reducing the shortest path and incurring an immediate cost equal to that reduction. The minimum cost reduction between two consecutive states is $\delta^{\alpha-1}-\delta^{\alpha}$, in the case when $j<\alpha$. We can therefore conclude that
\begin{equation}
\sum_{i=1}^\nsteps E[w_i] \geq -\gamma \cdot\nsteps/2\cdot (\delta^{a-1}-\delta^{\alpha})=\Omega(\nsteps)
\end{equation}
as required. \qed
\end{proof}

\subsection{Monotone distance}
\label{subsec:monotone}
We next consider the class of task graphs with the monotone distance property.  Recall that for such task graphs, the distance from any given state to $t$ is monotonically non-increasing along every edge in the graph. 

We first note that if $\dist$ is the present-bias distribution, and there exists some $\delta > 0$ such that $\dist(1+\delta) = 0$ (that is, $1+\delta$ is a lower bound on the present bias parameter), then there exist task graphs with the monotone distance property with procrastination ratio $\Omega(\delta \nsteps)$.  Indeed, consider the ski-renting example from Section \ref{sec:examples}: if $\dist$ places all of its mass on values greater than $\tfrac{1}{1-\delta}$, the agent would choose to pay $\delta$ at each step.  
This leads to a total cost of $\nsteps \delta$, whereas a cost of $1$ was possible.

This example motivates us to consider distributions $\dist$ that satisfy a mild condition, which essentially captures the requirement that there be sufficient mass in neighborhoods around $\bias = 1$.  Specifically, we will require that there exist some constants $\beta, \delta > 0$ such that $\dist(x)\ge \beta\cdot(x-1)$ for all $x\in[1,1+\delta]$.\footnote{Sufficient conditions for our requirement including having a positive mass at $\bias = 1$, or having a constant lower bound on the density function in any neighborhood around $\bias = 1$.}
%
%
%
We show that under this condition, the procrastination ratio is bounded by a constant.

\begin{theorem}
\label{thm:msp}
For any $\beta, \delta > 0$, if $\dist(x)\ge \beta\cdot(x-1)$ for all $x\in[1,1+\delta]$, then any task graph with monotone
distances has procrastination ratio at most $\max(1+\frac{1}{\beta},\frac{1}{\beta\delta})$. 
\end{theorem}
\begin{proof}
As in Theorem~\ref{th:optimal_contract}, we will characterize the task graphs with maximum procrastination ratio by way of an analogy with single-parameter auctions.  Following the notation in the proof of Theorem~\ref{th:optimal_contract}, we have that the problem of maximizing the procrastination ratio can be expressed recursively (up to a constant $d(v_i,t_{n+1})$) as 
$$\max_{V_{i+1}}\Big(1-E_{\biasi\sim\dist}[x(v^*_{i+1})]+E_{\biasi\sim\dist}[p(v^*_{i+1})]\opt(i+1,1)\Big).$$
The monotone distances property imposes the constraint that $\dtot[v_{i+1}] \leq \dtot[v_i]$ for all $v_{i+1} \in V_{i+1}$.  Since $p(v_{i+1}) = \tfrac{\dtot[v_{i+1}]}{\dtot[v_i]}$, this translates to a constraint that $p(v^*_{i+1}) \in [0,1]$.  Unlike the proof of Theorem~\ref{th:optimal_contract}, we cannot employ Theorem~\ref{thm:optpricing}, since the price suggested by Theorem~\ref{thm:optpricing} may be greater than $1$.
Instead, we will use the following variation.

\begin{claim}
\label{thm:optpricing.bounded}
Let $M\subseteq[0,1]\times [0,1]$ be a menu of options $(x,p)$ such that $(0,0)\in M$.  For any such $M$ and $v\in\R_{\ge 1}$, write $(x_M(v),p_M(v))=\argmax_{(x,p)\in M}[vx-p]$.  Then for all distributions $\dist$ with support on $[1,\infty)$, and for all $\alpha\in\R,\beta\in\R_{\ge 0}$, there exists an $x^* \in [0,1]$ and corresponding $M^*=\{(0,0),(x^*,1)\}$ such that
$$\alpha E_{v\sim\dist}[x_M(v)]+\beta E_{v\sim\dist}[p_M(v)]\leq\alpha E_{v\sim\dist}[x_{M^*}(v)]+\beta E_{v\sim\dist}[p_{M^*}(v)].$$
\end{claim}

\begin{proof}
For simplicity we will prove the claim under the assumption that menu $M$ is finite.  We note that the result directly extends to general $M$.  Also, in the context in which we apply the claim (namely, when menu $M$ corresponds to a task graph), $M$ is finite.

Write $M = \{(x_1, p_1), \dotsc, (x_k, p_k)\}$.  We will say that the item $(x_i, p_i)$ is \emph{chosen} by value $v$ if $(x_M(v),p_M(v)) = (x_i, p_i)$.  We will partition the set of possible values $v$ (i.e., $\R_{\ge 1}$) according to their chosen lottery.  Note that, for each $i$, the set of values that choose $(x_i, p_i)$ forms an interval, say $[v_{i-1}, v_i)$.  Without loss of generality, we can assume every element of $M$ is chosen by some value.  In this case, we can choose to index lotteries so that $1 = v_0 < v_1 < \dots < v_k = \infty$.  In this case, we must have $0=x_0\le x_1<x_2<\dots<x_k\le 1$ and $p_1<p_2<\dots<p_k\le 1$. We note that the latter can be derived from well-known payment identities of incentive compatible auctions~\cite{Myerson}.
Moreover, the fact that $v_{i-1}$ is the threshold type between choosing $(x_{i-1}, p_{i-1})$ and $(x_i, p_i)$, and is therefore indifferent between these auctions, implies that
\begin{align}
v_{j-1}\cdot(x_j-x_{j-1})=p_j-p_{j-1}\quad & \text{if } j\in\{2,\dots,k\} \nonumber\\
v_0\cdot x_1-p_1>0 \quad & \text{if } j=1.
\label{eq:app_incentive_constraint}
\end{align}

We now consider a collection of two-element menus $M_1,\dots,M_k$, where $M_j=\{(0,0),(\frac{1}{v_{j-1}},1)\}$. 
We will construct a distribution over these menus, and consider offering one of them at random: say menu $M_j$ is chosen with probability $m_j$.
We want to choose $\{m_1,\dots,m_k\}$ such that for every value $v\in[v_{j-1},v_j]$, the expected allocation chosen by $v$ under this distribution of menus is exactly $x_j$ (the allocation under $M$), and the expected payment is at least $p_j$. Note then that this random menu, call it $M^r$, can only increase our objective of interest:  
$$\alpha E_{v\sim\dist}[x_M(v)]+\beta E_{v\sim\dist}[p_M(v)]\leq\alpha E_{v\sim\dist}[x_{M^r}(v)]+\beta E_{v\sim\dist}[p_{M^r}(v)].$$
Since $M^r$ is a convex combination of menus $M_j$, there must exist at least one $M_j$ for which the objective is at least as high as for $M^r$.  This $M_j$ satisfies the requirements of the claim. 

It remains to find appropriate probabilities $m_1,\dots, m_k$.  We will choose $m_j=v_{j-1}\cdot(x_j-x_{j-1})$, where $j\in\{1,\dots,k\}$ and $x_0=0.$ 

We observe that each type $v\in(v_{i-1},v_i)$ 
chooses the option $(\frac{1}{v_{j-1}},1)$ in $M_j$ for every $j\in\{1,\dots,i\}$ and option $(0,0)$ for the remaining $j$. 
Therefore, the allocation probability and expected payment of type $v$ is
\begin{align*}
x(v)= &\sum_{j=1}^{i} \frac{1}{v_{j-1}} \cdot m_j  = \sum_{j=1}^{i} x_j-x_{j-1}=x_i\\
p(v)= & \sum_{j=1}^{i} m_j  = \sum_{j=1}^{i} v_{j-1}\cdot (x_j-x_{j-1})  = v_0\cdot x_1 + \sum_{j=2}^{i} (p_j-p_{j-1})= p_i-p_1+x_1\ge p_i,
\end{align*}
where the last inequality and previous two equalities follow from \eqref{eq:app_incentive_constraint} and telescopic summation.
Since $p(v_k)\le 1$ we also get that $\sum_{j=1}^{k} m_j\le 1$.  If $\sum_{j=1}^{k} m_j$ is strictly less than $1$, we 
can extend to a probability distribution by adding a dummy lottery $M_0=\{(0,0)\}$ that is offered with probability $1 - \sum_{j=1}^{k} m_j$.
\end{proof}


Claim~\ref{thm:optpricing.bounded} implies that there is an optimal menu $V_{i+1}=\{v_{i+1},u_{i+1}\}$ where $x(v_{i+1}) \leq 1$, $p(v_{i+1}) \leq 1$, $x(u_{i+1})=0$, and $p(u_{i+1})=0$.  Given a partially constructed task graph with $V_i=\{v_i,u_i\}$, we can therefore optimally extend it by defining $V_{i+1}=\{v_{i+1},u_{i+1}\}$.  Applying the same reasoning as in the proof of Theorem~\ref{th:optimal_contract}, we can conclude that the transitions from $v_i$ to $V_{i+1}$, and the continuation shortest path weights, will be as follows:
\begin{enumerate}
\item $w(v_iu_{i+1})=\disti{t}$ and $\dtot[u_{i+1}]=0$, and
\item $w(v_iv_{i+1}) \leq \dtot[v_i]$ and $\dtot[v_{i+1}]\le \disti{t}$.
\end{enumerate}

Note that the task graph is now fully specified, except for the values of $w(v_iv_{i+1})$ and $\dtot[v_{i+1}]$ for each $i \leq n$. (Recall that we must have $v_{n+1} = u_{n+1} = t$.)  For any such graph and distribution $\dist$, we can write $q_i$ for the probability that the agent would choose $u_{i+1}$ instead of $v_{i+1}$, from $v_i$.  That is,
\[ q_i=\Prx[\biasi\sim\dist]{\biasi\cdot \frac{w(v_iv_{i+1})}{\dtot[v_i]}+\frac{\dtot[v_{i+1}]}{\dtot[v_i]}\ge\biasi\cdot 1}. \]
We can write the following expression for the optimum at each state $v_i$ scaled by $\dtot[v_i]$, again following the notation from the proof of Theorem~\ref{th:optimal_contract}:
\be
\label{eq:opt_monotone}
\opt(i,1)=q_i \cdot 1 + (1-q_i)\left(\frac{w(v_iv_{i+1})}{\dtot[v_i]}+\frac{\dtot[v_{i+1}]}{\dtot[v_i]}\cdot\opt(i+1,1)\right).
\ee
To complete the proof of Theorem~\ref{thm:msp}, it suffices to show that $\opt(i,1)\le\max(1+\frac{1}{\beta},\frac{1}{\beta\delta})$ for each $1 \leq i \leq  n$.  We will show this by backward induction on $i$.  The base case $i = n$ follows immediately from the fact that $OPT(n, 1) = 1 \le\max(1+\frac{1}{\beta},\frac{1}{\beta\delta})$, since the agent has only one option: moving to vertex $t$.

Consider $i < n$.  Let $\Delta\eqdef\max(1+\frac{1}{\beta},\frac{1}{\beta\delta})$. By \eqref{eq:opt_monotone}, it is sufficient to show that 
$\Delta\ge q_i + (1-q_i)(\frac{w(v_iv_{i+1})}{\dtot[v_i]}+\frac{\dtot[v_{i+1}]}{\dtot[v_i]}\cdot\Delta)$, or equivalently that
\be
\label{eq:Delta_monotone}
\Delta q_i +\Delta(1-q_i)\left(1-\frac{\dtot[v_{i+1}]}{\dtot[v_i]}\right)\ge q_i+(1-q_i)\frac{w(v_iv_{i+1})}{\dtot[v_i]}.
\ee
We consider two cases based on the value of $q_i$. 

\noindent
\textbf{Case 1: $q_i\ge\beta\delta$. } Then $\Delta\cdot q_i\ge\frac{q_i}{\beta\delta}\ge 1$, and hence the left hand side of \eqref{eq:Delta_monotone} is at least $1$.  On the other hand, since $\frac{w(v_iv_{i+1})}{\dtot[v_i]}\le 1$, the right hand side of \eqref{eq:Delta_monotone} is at most $1$.

\noindent
\textbf{Case 2: $q_i<\beta\delta$. } We observe that 
$q_i=\dist\left(\frac{\dtot[v_{i+1}]}{\dtot[v_i]-w(v_iv_{i+1})}\right)$ 
and by the condition of the theorem 
that $F(x)\ge\beta(x-1)$ for $x\in[1,1+\delta]$ we therefore have $q_i\ge\beta\cdot\left(\frac{\dtot[v_{i+1}]}{\dtot[v_i]-w(v_iv_{i+1})}-1\right)$.
We therefore have 
\begin{align*}
& \Delta q_i +\Delta(1-q_i)\left(1-\frac{\dtot[v_{i+1}]}{\dtot[v_i]}\right) \\
& \qquad \geq q_i+\frac{1}{\beta}q_i+(1-q_i)\left(1-\frac{\dtot[v_{i+1}]}{\dtot[v_i]}\right)  \\
& \qquad \ge q_i+ (1-q_i)\left(\frac{1}{\beta}q_i + 1-\frac{\dtot[v_{i+1}]}{\dtot[v_i]}\right)  \\
&\qquad \ge q_i+ (1-q_i)\left(\frac{\dtot[v_{i+1}]}{\dtot[v_i]-w(v_iv_{i+1})}-1 + 1-\frac{\dtot[v_{i+1}]}{\dtot[v_i]}\right)\\
& \qquad  = q_i+ (1-q_i)\frac{\dtot[v_{i+1}]}{\dtot[v_i]-w(v_iv_{i+1})}\frac{w(v_iv_{i+1})}{\dtot[v_i]}\\
& \qquad \geq q_i + (1-q_i)\frac{w(v_iv_{i+1})}{\dtot[v_i]},
\end{align*}	
where the first inequality follows from $\Delta \geq 1 + \frac{1}{\beta}$, the second inequality we simply reduced the second term by $(1-q_i)$; in the third inequality we applied our lower bound on $q_i$. 
The final inequality follows because $\dtot[v_{i+1}]+w(v_iv_{i+1})\ge \dtot[v_i]$, which implies that $\frac{\dtot[v_{i+1}]}{\dtot[v_i]-w(v_iv_{i+1})}\ge 1$. We have therefore established \eqref{eq:Delta_monotone}, as desired. \qed
\end{proof}

\section{Conclusions and Open problems}

We have examined how variability in an individual's decision making process can affect behavior in 
time-inconsistent planning scenarios. We adopted the graphical framework of Kleinberg and Oren~\cite{KleinbergO14}
and characterized worst-case task graphs for the agent. 
We showed that depending on the distribution
of the present-bias parameter, the worst-case procrastination ratio is either bounded by constant, is
at most linear, or grows at least
exponentially with the number of agent's steps $\nsteps$. We also examined two natural families 
of tasks: (i) those in which the cost of reaching the goal is never more than at the outset, 
and (ii) those in which an agent can never lose progress toward the goal. 
We showed that in the first scenario the worst-case procrastination ratio is at most $O(\nsteps)$, and this is tight, 
and for the second scenario the procrastination ratio is $O(1)$ under some mild assumptions on the 
present-bias distribution. 

Our work leaves open many avenues for future study.  The following are a few concrete questions:

\begin{enumerate}
\item We always assume that $\bias\ge 1$. What could be the worst-case procrastination ratio if $\biasi$ can be smaller than $1$?
      For example, an agent might receive a reminder about the task and becomes anxious to complete it as soon as possible.
\item We examine worst-case procrastination ratio for a fixed distribution of $\bias\sim\dist$. Is there a bound
      on procrastination ratio for a given pair of the task graph and distribution $\dist$?  In other words, are certain distributions better suited to certain tasks than others?  What is an important intrinsic parameter of 
			a weighted graph that can lead to large procrastination? 
\item We examined the case of a simple, na\"{i}ve agent who never anticipates their time-inconsistency. What can one say about agents
      who can reason about their future propensity to procrastinate? 
\end{enumerate}
 


\bibliographystyle{ACM-Reference-Format-Journals}
\bibliography{bib}
\end{document}